\documentclass[runningheads,orivec,a4paper]{llncs}
\usepackage{amssymb}
\usepackage{dsfont}
\usepackage{graphicx}
\usepackage{mathtools}
\usepackage{physics}
\usepackage{textcomp}
\usepackage{amsmath}
\usepackage{amsopn}
\usepackage{tabularx,booktabs,subcaption}
\usepackage{hyperref}
\usepackage{array}
\newcolumntype{Y}{>{\centering\arraybackslash}X}
\newcolumntype{C}[1]{>{\centering\let\newline\\\arraybackslash\hspace{0pt}}m{#1}}

\DeclareMathOperator*{\argmin}{\arg\!\min}
\DeclareMathOperator*{\supp}{supp}

\begin{document}
\mainmatter  
\title{Incompressible image registration \texorpdfstring{\\}{} using divergence-conforming B-splines}
\titlerunning{Incompressible registration using divergence-conforming B-splines}
%
\author{Lucas Fidon\inst{1}
    \and Michael Ebner\inst{1,2}
	\and Luis C. Garcia-Peraza-Herrera\inst{1,2}
	\and \\Marc Modat\inst{1}
	\and S\'ebastien Ourselin\inst{1}
	\and Tom Vercauteren\inst{1}}
\authorrunning{Lucas Fidon et al.}
%
\institute{School of Biomedical Engineering \& Imaging Sciences, King's College London, UK \and
University College London, UK}

\maketitle              
\begin{abstract}
Anatomically plausible image registration often requires volumetric preservation.
%
%
Previous approaches to incompressible image registration have exploited relaxed constraints, ad hoc optimisation methods or practically intractable computational schemes.
Divergence-free velocity fields have been used to achieve incompressibility in the continuous domain, although, after discretisation, no guarantees have been provided. 
%
%
%
In this paper, we introduce stationary velocity fields (SVFs) parameterised by divergence-conforming B-splines in the context of image registration.
We demonstrate that sparse linear constraints on the parameters of such divergence-conforming B-Splines SVFs lead to being exactly divergence-free at any point of the continuous spatial domain.
In contrast to previous approaches, our framework can easily take advantage of modern solvers for constrained optimisation, symmetric registration approaches, arbitrary image similarity and additional regularisation terms.
%
%
We study the numerical incompressibility error for the transformation in the case of an Euler integration, which gives theoretical insights on the improved accuracy error over previous methods.
We evaluate the proposed framework using synthetically deformed multimodal brain images, and the STACOM'11 myocardial tracking challenge.
Accuracy measurements demonstrate that our method compares favourably with state-of-the-art methods whilst achieving volume preservation.

\end{abstract}

\section{Introduction}
Medical image registration 
%
consists of finding a spatial transformation that aligns two or more images. 
The intrinsic ill-posedness of registration can lead to anatomically implausible transformations
associated with unrealistic volumetric distortion of the anatomy.
%
For certain anatomical regions, such as the myocardium, physiologically plausible image registration requires volumetric preservation which corresponds to so-called \emph{incompressible} registration~\cite{Bistoquet2008,Mansi2011,Tobon-Gomez2013}.


%
%
Incompressible medical image registration between a pair of images $I_1, I_2$ 
can be defined as a constrained optimisation problem:
\begin{equation}
    \label{eq:reg_incompressible_continuous}
        \argmin_{\Phi \in \mathcal{D}\left(\Omega\right)}\,\,  
        \mathcal{L}(I_1, I_2, \Phi) + R(\Phi) \quad
        \text{s.t.} \quad
        \left[
        \forall x \in \mathcal{M},\quad \det\left(J_{\Phi}(x)\right) = 1
        \right]
\end{equation}
where $\Omega$ is the image domain, $\mathcal{M} \subset \Omega$ is the incompressible region, $\mathcal{L}$ an image dissimilarity measure, $R$ a regularisation term, $\mathcal{D}\left(\Omega\right)$ the group of diffeomorphic transformations from $\Omega$ onto itself, and $J_{\Phi}$ the Jacobian matrix of the transformation $\Phi$.

In practice, to solve \eqref{eq:reg_incompressible_continuous}, the transformation must be parameterised by a finite number of parameters, and the constraint $\forall x \in \mathcal{M},\, \det\left(J_{\Phi}(x)\right) = 1$ must be reduced to a finite number of equality constraints.
There are two approaches to discretise the constraint:
1) relaxing the constraint into an additive soft constraint~\cite{Aganj2015,DeCraene2012,Heyde2016,Rohlfing2003}.
2) using a specific parameterisation for the transformation~\cite{Bistoquet2008,Mang2018,Mansi2011}.
%
%
Soft constraints introduce hyperparameters that are difficult to tune reliably and cannot guarantee an incompressible transformation.
%
In \cite{Bistoquet2008}, the transformation is parameterised by a divergence-free displacement. Yet, this approach only provides a first order approximation of an incompressible deformation, still requiring a soft constraint.
%
%
Recently \cite{Mang2018,Mansi2011} proposed 
to parameterise the transformation by a divergence-free stationary velocity field (SVF).
The transformation is obtained via the Lie exponential mapping $\exp: v \mapsto \Phi $ that maps any divergence-free SVF $v$ into an incompressible transformation $\Phi$~\cite{Mang2018,Mansi2011}.
This reduces the non-convex constraint in \eqref{eq:reg_incompressible_continuous} to a linear constraint in the continuous domain.
%
However, in~\cite{Mang2018,Mansi2011} the SVF is parameterised as linear B-splines, and the linear constraint is discretised by imposing the constraint only on the points of the deformation grid.
%
As a result, guarantees to obtain a continuous incompressible transformation do not hold anymore.
%
%
To mitigate this issue,~\cite{Mang2018} proposed to work with images of higher resolution with a finer grid for the linear B-spline leading to the need for distributed super-computing.
Moreover,~\cite{Mang2018,Mansi2011} methods are limited to using sum of squared differences (SSD) as image similarity metric, which limits their applicability to images with similar intensity distributions.

In this paper, we propose a constrained optimisation framework for incompressible diffeomorphic registration that allows to use any smooth image similarity and regularisation penalty.
As an efficient means of discrete SVF parameterisation for this problem, we introduce multivariate divergence-conforming B-splines that have recently raised interest in computational physics~\cite{Evans2013a}. 
We demonstrate that their properties can be exploited to impose bounds on the divergence of the SVF over the entire continuous space using sparse linear constraints on its finite parameters.
Our general problem formulation allows us to solve incompressible registration using any state-of-the-art optimiser for constrained non-convex optimisation (e.g. \texttt{IPOPT}~\cite{Wachter2006}).
%
For evaluation, we initially apply our method for multi-modal incompressible registration of synthetically deformed brains.
We then compare our approach against the state-of-the-art results published for the STACOM'11 myocardial tracking challenge dataset~\cite{Tobon-Gomez2013} and achieve similar results while better retaining the incompressibility.
%

\newpage
\section{Method}
\subsubsection{Divergence-conforming B-splines}
%
%
General B-splines are very popular to parameterise deformations and velocity fields
over a continuous spatial domain $\Omega$
with a finite number of parameters
$\phi_i \in \mathds{R}^{3}$:
\begin{equation}
    \label{eq:bspline3d}
     \forall \left(x, y, z\right) \in \Omega,\quad
     v\left(x,y,z\right) = \sum_{i} B^k_{i,X}(x) B^k_{i,Y}(y) B^k_{i,Z}(z)\,\phi_{i},
\end{equation}
where 
$(x_i, y_i, z_i)$ are the knots of a regular grid of spacing $\delta x \times \delta y \times \delta z$ and 
$B^k_{i,U}$ are the
1D B-spline basis functions of order $k$ in the direction $U\in \{X, Y, Z\}$ (see supplemental material for more details).
%
Popular choices of the order are $k=1$ (linear B-splines)~\cite{Mang2018,Mansi2011} and $k=3$ (cubic B-splines)~\cite{Tobon-Gomez2013}.
%
A fundamental property of 1D B-spline basis functions of the variable $u$, on a regular grid 
of spacing $\delta u$ with $n$ knots is that:
\begin{equation}
\label{eq:deriv_bsplines}
    \forall i,
    \quad \dv{B^k_{i,U}}{u} = \frac{B^{k-1}_{i,U} - B^{k-1}_{i+1,U}}{\delta u}
\end{equation}
This implies that the derivative of a 1D B-spline on a regular grid is also a 1D B-spline on the same grid, albeit of lower order.
However, this property does not extend to 3D B-splines using definition \eqref{eq:bspline3d}. 
Especially, the divergence of those 3D B-splines are not B-splines,
because of the mixed order appearing with first partial derivatives.
%
This limitation makes it more difficult to relate properties of the velocity field $v$ to the values of the parameters
$\phi_i$.

To overcome this limitation, we propose to use a 3D divergence-conforming B-splines~\cite{Evans2013a} to parameterise $v$.
The orders of the 3D B-splines basis of each component are chosen so that the divergence of $v$ is, in a continuous sense, exactly a 1D B-spline of order $k$. 
Using the same notations as in \eqref{eq:bspline3d}, and 
using $\phi_i = (\phi^X_i,\phi^Y_i,\phi^Z_i)$,
we have:
\begin{equation}
    \forall (x,y,z) \in \Omega,\quad 
    v(x,y,z) = 
    \left(
    \begin{array}{c}
         \sum_{i} B^{k+1}_{i,X}(x) B^k_{i,Y}(y) B^k_{i,Z}(z)\,\phi^X_{i}  \\
         \sum_{i} B^k_{i,X}(x) B^{k+1}_{i,Y}(y) B^k_{i,Z}(z)\,\phi^Y_{i} \\
         \sum_{i} B^k_{i,X}(x) B^k_{i,Y}(y) B^{k+1}_{i,Z}(z)\,\phi^Z_{i}
    \end{array}
    \right)
\end{equation}
Using \eqref{eq:deriv_bsplines}, we obtain the continuous divergence $\div{v}$ for all $(x,y,z) \in \Omega$:
\begin{equation}
\label{eq:div_divergence-conforming}
    \begin{split}
        \div{v}(x,y,z) =& \sum_{i} B^{k}_{i,X}(x) B^k_{i,Y}(y) B^k_{i,Z}(z)\, \psi_i \\
        \text{s.t.}\,\, \forall i,\quad \psi_i =& \,
    \frac{\phi^X_{i} - \phi^X_{i-1}}{\delta x} +
    \frac{\phi^Y_{i} - \phi^Y_{i-1}}{\delta y} +
    \frac{\phi^Z_{i} - \phi^Z_{i-1}}{\delta z}
    \end{split}
\end{equation}
As a consequence, the following lemma states that $\div{v}$ is uniformly bounded at any point of a continuous subregion $\mathcal{M} \subset \Omega$ provided that a finite number of linear constraints are satisfied by the coefficients 
$(\phi^X_i,\phi^Y_i,\phi^Z_i)$.

\begin{lemma}
\label{lemma:1}
Let $k\geq 2$, and $\mathcal{M}$ be a non-empty subset of $\Omega$. Let $\epsilon \geq 0$
and let $\mathbf{J}_{\mathcal{M}}=\{i
\,\,|\,\, 
\left(
\supp B^{k}_{i,X} \times \supp B^{k}_{i,Y} \times \supp B^{k}_{i,Z}
\right)
\cap \mathcal{M} \neq \emptyset
\}$.
If
\[
    \forall i \in \mathbf{J}_{\mathcal{M}},\quad -\epsilon \leq 
    \frac{\phi^X_{i} - \phi^X_{i-1}}{\delta x} +
    \frac{\phi^Y_{i} - \phi^Y_{i-1}}{\delta y} +
    \frac{\phi^Z_{i} - \phi^Z_{i-1}}{\delta z}
    \leq \epsilon,
\]
then at any point $m \in \mathcal{M}$, it holds that
$\,\abs{\div{v}(m)} \leq \epsilon$.
\end{lemma}

\begin{proof}
The proof follows from \eqref{eq:div_divergence-conforming} and the fact that the value of a B-spline at any point $m$ is bounded by the values of the B-spline coefficients associated to the knots that are close to $m$ (see supplementary material for more details).
\end{proof}

\subsubsection{Optimisation formulation and implementation}
In this section, we formulate the optimisation problem for diffeomorphic registration with the proposed incompressibility constraint on a subregion $\mathcal{M}$.
%
Using Lemma~\ref{lemma:1} and previous notations, our (symmetric) optimisation formulation of \eqref{eq:reg_incompressible_continuous} is:
\begin{equation}
    \label{eq:optim_formulation}
    \begin{split}
        \argmin_{\Theta=\{(\phi^X_i,\phi^Y_i,\phi^Z_i)\}_{i}} \,\, & \mathcal{L}(I_1 \circ \widetilde{\exp}(v(\Theta)), I_2) + \mathcal{L}(I_1, I_2 \circ \widetilde{\exp}(-v(\Theta))) + R(v(\Theta))\\
        \text{s.t. } \quad &
        \forall i \in \mathbf{J}_{\mathcal{M}},\,\,
            \frac{\phi^X_{i} - \phi^X_{i-1}}{\delta x} +
            \frac{\phi^Y_{i} - \phi^Y_{i-1}}{\delta y} +
            \frac{\phi^Z_{i} - \phi^Z_{i-1}}{\delta z}
            = 0\\
    \end{split}
\end{equation}
where 
$\widetilde{\exp}$ is an approximation
of the Lie exponential.
%
Thus we obtain a constrained optimisation formulation that guarantees the SVF to be \emph{exactly} divergence-free over the entire continuous subregion $\mathcal{M}$ and that can be solved with efficient state-of-the-art optimisers.
%
Using state-of-the-art solvers like \texttt{IPOPT}, that uses a primal-dual interior-point filter line-search method, an approximated solution of \eqref{eq:optim_formulation} that satisfies the constraints up to machine precision can be obtained.

%

\subsubsection{Lie exponential approximation and incompressibility}
%
%
As the Lie exponential has to be approximated in practice, having a divergence-free velocity field does not strictly guarantee that the resulting deformation will be incompressible.
We now quantify this approximation in our framework with respect to the time step $\tau$ for the Euler method.
Let $v$ denote a (SVF) solution of \eqref{eq:optim_formulation} that fulfills the divergence-free constraint up to machine precision $\epsilon_{mach}$. For any point $m \in \mathcal{M}$, the Jacobian of the first step of the Euler method, $I + \tau v$, fulfills:
\begin{equation}
    \label{eq:jac_one_step_error}
    \det\left(J_{I + \tau v}(m)\right) = 1 + \tau \div{v}(m) + \mathcal{O}\left(\tau^2\right)
    \quad \text{where} \quad \abs{\div{v}(m)} \leq \epsilon_{mach}.
\end{equation}
Then, by composing $1/\tau$ times, if the point $m$ remains inside $\mathcal{M}$ during the Euler integration, the Jacobian of $\widetilde{\exp}(v)$ satisfies (see supplementary material):
\begin{equation}
    \label{eq:jac_error}
    \det\left(J_{\widetilde{exp}(v)}(m)\right) = 1 + \mathcal{O}\left(\tau + \epsilon_{mach}\right)
\end{equation}
This approximation is independent to the spatial spacing and only depends on the time integration step $\tau$.
This is, to the best of our knowledge, a new state-of-the-art approximation for an incompressible diffeomorphic deformation parameterised by an SVF.
However, it is worth noting that \eqref{eq:jac_error} is only guaranteed
if $m$ remains inside $\mathcal{M}$ during the Euler integration.
Thus, when $\mathcal{M}$ is not equal to the entire spatial domain, this approximation may not hold for large deformations. 


\section{Evaluation on synthetic data for incompressible multi-modal registration}


\begin{table}[tb]
	\centering
	\caption{\textbf{Validation on incompressible multi-modal registration.} 
	RMSE values between predicted and ground-truth transformations are reported in mm. MAE($\abs{\text{Jacobian} - 1}$) stands for the Mean Absolute Error between the Jacobian map of the predicted transformation and a map uniformly equal to $1$ (perfect incompressibility).
	For both metrics, we reported mean value (standard deviation) over 60 registrations.
	}
	\begin{tabularx}{\textwidth}{c *{4}{Y}}
		\toprule
		 \textbf{Error measures} & Affine & Cubic B-splines & Ours\\ 
		\midrule
		RMSE (transformation)      & 2,28 (0,68) & 0.96 (0.44) & \textbf{0.90 (0.45)}\\
		MAE($\abs{\text{Jacobian} - 1}$)  & 0.0062 (0.005) & 0.054 (0.02) & \textbf{0.00079 (0.0004)}\\

		\bottomrule
	\end{tabularx}
	\label{tab:multimod_results}
\end{table}

We start by evaluating our method on synthetic incompressible registration in the context of multi-modal MRI data.

\subsubsection{Data generation method} We used T1, T2 and PD brain images from the IXI dataset\footnote{\url{https://brain-development.org/ixi-dataset/}}.
We generated realistic ground-truth incompressible transformations in two steps.
First, 
we non-linearly registered a pair of T1 images coming from different patients using a classical cubic B-splines SVF $v_0$.
Second, 
we generated a quasi-divergence-free SVF $v$ by projecting 
$v_0$
on the space of quasi-divergence-free SVFs.
This corresponds to solving:
\begin{equation}
    \label{eq:data_generation}
        \argmin_{\Theta=\{\phi_i\}_i} 
        \frac{1}{2}\norm{v(\Theta) - v_0}^2 \quad
        \text{s.t.} 
        \quad \forall i, 
        \quad \div{v}(x_i, y_i, z_i) = 0
\end{equation}
where $v$ and $v_0$ are parameterised with classical cubic B-splines as in \eqref{eq:bspline3d}. We note that any potential bias towards the space of divergence-conforming B-splines is avoided in this comparison. We used \texttt{IPOPT}~\cite{Wachter2006} to solve \eqref{eq:data_generation}.

\subsubsection{Evaluation}
We generated a ground-truth quasi-incompressible SVF for 10 patients using the previous generation procedure, taking as fixed images 10 other patients.
Then for a given subject and for any pair of imaging modalities $M_1, M_2 \in \{T1, T2, PD\}$, we warped the first image $I(M_1)$ using the inverse of the ground-truth transformation. The task consists in estimating $v_{GT}$ by registering $I(M_1) \circ \widetilde{\exp}(-v_{GT})$ to $I(M_2)$.
%
%
We compared the proposed method with divergence-con\-forming B-splines SVF under incompressibility constraint to a classic registration approach based on cubic B-splines SVF similar to the one used to generate the ground-truth.
We also performed an affine registration to illustrate that the ground-truth transformation is not just affine.

\subsubsection{Implementation details}
We used \texttt{NiftyReg}~\cite{Modat2010} for the diffeomorphic registrations using cubic B-splines SVF.
For both \texttt{NiftyReg} and our incompressible implementation we used the same hyperparameters (grid size 5mm, $3$ levels of pyramid) and objective function, consisting of Normalised Mutual Information (NMI) as a similarity measure (weight $0.95$) and a bending energy regularisation term (weight $0.05$). 
The optimiser differs, as we used \texttt{IPOPT} optimiser while a Conjugate Gradient approach is used in \texttt{NiftyReg}.

\subsubsection{Results} \texttt{NiftyReg} and our approach recovered the ground-truth SVF up to a RMSE of the order of the images resolution, as shown in Table.~\ref{tab:multimod_results}. In addition, our method recovered the incompresibility with a higher accuracy.


\begin{figure}[tb!]
    \centering
    \includegraphics[width=\textwidth]{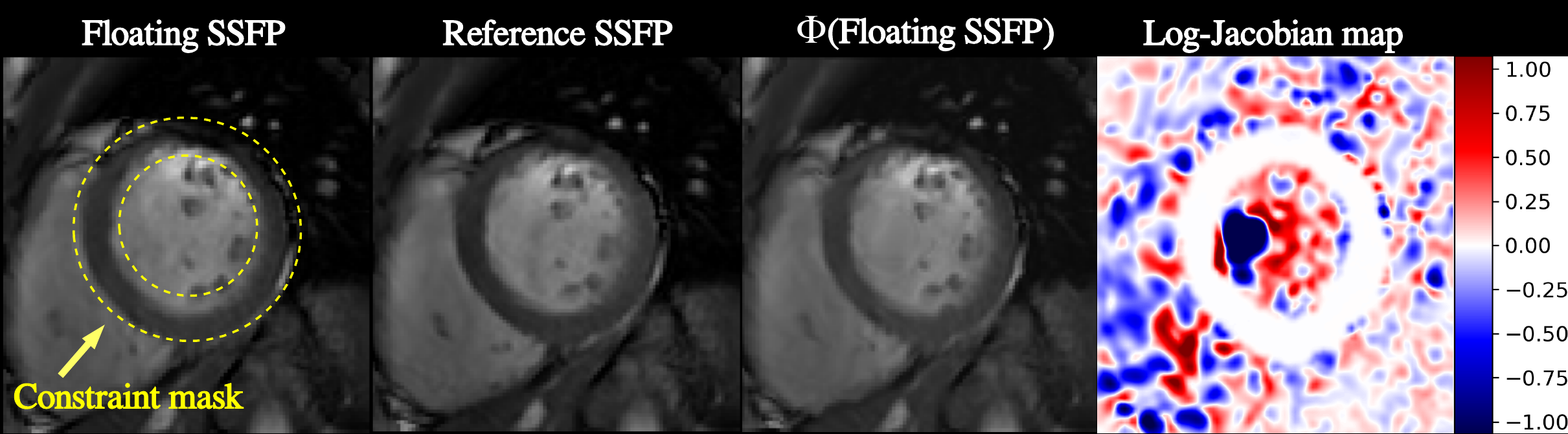}
    \caption{We aim to register SSFP images at different times of the cardiac cycle.
    During the registration, the myocardium should not compress or expand.
    To impose this constraint we provide a mask of it to our algorithm (dashed in yellow).
    The floating image is warped into $\Phi$(Floating SSFP) so that it matches the reference image, while keeping the deformation of the myocardium incompressible, as shown in the log-Jacobian map.
    Red and blue represent expansion and compression respectively.
    }
    \label{fig:logJacobian}
\end{figure}

\section{Evaluation for myocardial tracking}
We evaluate the proposed method on the STACOM'11 myocardial tracking challenge~\cite{Tobon-Gomez2013}.
In particular, this allows a direct comparison of our framework with iLogDemons\cite{Mansi2011} a state-of-the-art incompressible registration method.

\subsubsection{Data}
The STACOM'11 dataset\footnote{\url{http://stacom.cardiacatlas.org/motion-tracking-challenge/}} contains 4D cine Steady State Free Precession (SSFP) of a full cardiac cycle with $30$ time points for $15$ patients.
Those SSFP come with
$12$ manually tracked landmarks.
We used only 13 of the 15 patients available because we found a shift between the images and the landmarks for two of them which we reported to the organisers.
%
The coordinates of the landmarks obtained by competitive methods of the challenge are available alongside the original data. This allows a direct comparison with 
our method.

\subsubsection{Implementation details} 
We used a similar implementation to the one of iLogDemons for this challenge, where we registered the first frame to all subsequent frames for each patient and used a manually delineated mask for the first frame. Yet, we used Local Normalised Cross Correlation as a similarity measure.
We used a bending energy regularisation, and a grid size of $3$mm.

\begin{figure}[bt!]
    \centering
    \begin{subfigure}[t]{0.37\textwidth} 
        \centering
        \includegraphics[width=\linewidth]{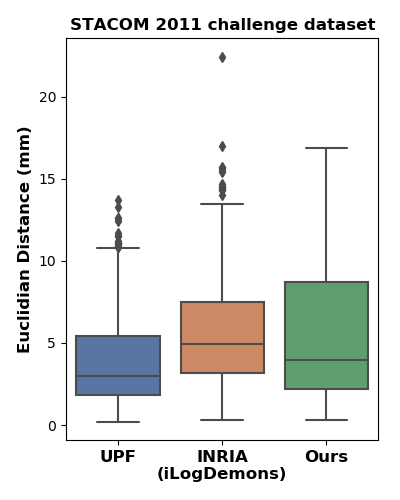}
    \end{subfigure}
    \begin{subfigure}[t]{0.62\textwidth} 
        \centering
        \includegraphics[trim={0 0 0 0.8cm},clip, width=\linewidth]{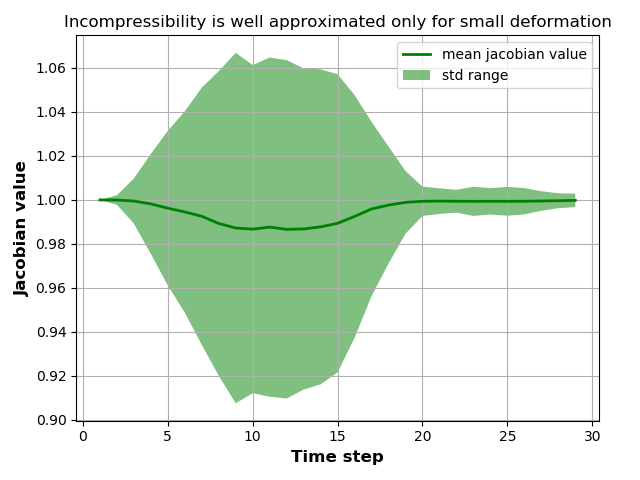} 
    \end{subfigure}
    \caption{(left) Distances to the landmarks after registering SSFP data.
    Our method achieves similar results to iLogDemons~\cite{Mansi2011} for an incompressible registration.
    UPF~\cite{DeCraene2012} achieves the highest accuracy, but it does not guarantee incompressibility.
    (right) Evolution of the Jacobian values distribution in the incompressible region (myocardium) while registering the first frame to all other time frames. The Jacobian values are uniformly close to 1 for small and moderate deformations of the myocardium. For large deformations (frames at the opposite to the cardiac cycle) 
    the mean Jacobian value is close to 1, but the dispersion of the Jacobian value distribution is large.
    This can be attributed to the use of a SVF.
    }
    \label{fig:ssfp_boxplot}
\end{figure}

\subsubsection{Results}
The evaluation of the myocardial tracking is based on the manually annotated landmarks at End Diastole and End Systole.
Using a Wilcoxon signed-rank test, we found that our results, shown in Fig.~\ref{fig:ssfp_boxplot}, are not statistically different from the results of iLogDemons~\cite{Mansi2011} in terms of landmark tracking error.
%
%
The log-Jacobian map of Fig.~\ref{fig:logJacobian} illustrates our approximation results \eqref{eq:jac_error}: for a small deformation, our registration framework can guarantee an incompressible deformation in a subregion with high accuracy.
While Fig.~\ref{fig:ssfp_boxplot} illustrates the degradation of this result for larger deformation, i.e. when registering frames at the opposite of the cardiac cycle.

\newpage

\section{Conclusion and Future Work}

\subsubsection{Limitations}
%
Similar to~\cite{Mang2018,Mansi2011}, the divergence-free constraint is imposed on the stationary velocity field (SVF) in a subregion $\mathcal{M}$, rather than imposing incompressibility on $\mathcal{M}$.
%
A voxel that is originally in $\mathcal{M}$ might exit $\mathcal{M}$ during the integration of the SVF and start being transported by velocity vectors that are not divergence-free (see supplemental material for a mathematical justification).
In this case, the deformation of this voxel is no longer incompressible.
As a result, using an SVF for incompressible registration applies only for deformations that are small enough or when the whole spatial domain is constrained to be divergence-free.
Investigating other diffeomorphic parameterisations is left for future work. 


\subsubsection{Advantages compared to previous methods}
Our method for incompressible registration relies on the parametrisation of the velocity field by a divergence conforming B-splines~\cite{Evans2013a}. 
In contrast to classical B-splines used in~\cite{Mang2018,Mansi2011}, it guarantees that the divergence of the velocity field is still a B-spline. 
This parametrisation along with constrained optimisation methods allows us to impose the velocity field to be divergence-free up to machine precision ($10^{-16}$ in our experiments). This is irrespective of the grid resolution chosen for the velocity field.
%
As a result, the proposed method 
is scalable to 3D images with high resolution.
%

We also proved an error bound for the incompressibility of the deformation for the proposed method in the case an Euler integration of the velocity field is used \eqref{eq:jac_error}.
The study of the error bound for other integrators, that may require additional interpolations (e.g. scaling-and-squaring), is left for future works.
Additionally, previous (quasi-)incompressible registration methods~\cite{Mang2018,Mansi2011} have been limited to using SSD as image similarity metric. We are proposing the first incompressible non-linear registration method that supports any smooth image similarity measure and spatial regularization.

\subsubsection*{Acknowledgments}
This project has received funding from the European Union's Horizon 2020 research and innovation programme under the Marie Sk{\l}odowska-Curie grant agreement TRABIT No 765148;
Wellcome [203148/Z/16/Z; 203145Z/16/Z; WT101957], EPSRC [NS/A000049/1; NS/A000050/1; NS/A000027/1; EP/L016478/1].
TV is supported by a Medtronic / RAEng Research Chair [RCSRF1819\textbackslash7\textbackslash34].

%
%
%
\bibliographystyle{splncs04.bst}
\bibliography{ms.bib}


\newpage
\begin{center}
\textbf{\Large --- Supplemental Document ---\\Incompressible image registration \\using divergence-conforming B-splines}\\[.6cm]
Lucas Fidon,$^{1}$ Michael Ebner,$^{1,2}$ Luis C. Garcia-Peraza-Herrera,$^{1,2}$ Marc Modat,$^{1}$ S\'ebastien Ourselin,$^{1}$ Tom Vercauteren$^{1}$\\[.3cm]
  \small ${}^1$School of Biomedical Engineering \& Imaging Sciences, King’s College London, UK\\
  ${}^2$University College London, UK\\
\end{center}

\setcounter{figure}{0}
\setcounter{table}{0}
\setcounter{page}{1}
\setcounter{section}{0}
\makeatletter

\section{B-splines notations}
To simplify the notations, let us assume that:
\begin{itemize}
    \item $\Omega=[0,1]^3 \subset \mathds{R}^3$ is the continuous spatial domain of the images,
    \item $k \in \mathds{N}$, the order of the B-splines basis, is less than $3$ (higher order values are rarely used in practice)
    \item $\left(\delta x, \delta y, \delta z\right)= (\frac{1}{n}, \frac{1}{n}, \frac{1}{n})$ is a regular spacing with $n > k$,
    \item $\Omega_{grid}^k$ is a regular grid of knots $\{(x_{i_X}, y_{i_Y}, z_{i_Z})=(\frac{i_X}{n},\frac{i_Y}{n},\frac{i_Z}{n})\}_{i_X,i_Y,i_Z=-k}^{n-1}$ for the spacing $\left(\delta x, \delta y, \delta z\right)$ on $[-\frac{k}{n},1-\frac{1}{n}]^3$.
\end{itemize}
%
We use the following notations for the B-splines basis functions of order $k$:
\begin{equation*}
    \begin{split}
    \forall i= (i_X, i_Y, i_Z) \in \{-k, \ldots, n-1\}^3, \,\,& \forall \left(x, y, z\right) \in [0,1]^3,\\
        B^k_{i_X,X}(x) =& B^k\left(\frac{x - x_{i_X}}{\delta x} - \frac{k+1}{2}\right)\\
        B^k_{i_Y,Y}(y) =& B^k\left(\frac{y - y_{i_Y}}{\delta y} - \frac{k+1}{2}\right)\\
        B^k_{i_Z,Z}(z) =& B^k\left(\frac{z - z_{i_Z}}{\delta z} - \frac{k+1}{2}\right)\\
    \end{split}
\end{equation*}
where the $B^k$ for $k \in \{0,1,2,3\}$ are defined by:
\begin{equation*}
    \label{eq:bspline}
    \begin{split}
        \forall t \in \mathds{R}, \quad & B^0\left(t\right) = \mathbf{1}_{[-\frac{1}{2}, \frac{1}{2}]}\left(t\right)\\
        & B^1\left(t\right) = \left(1 - |t|\right)\mathbf{1}_{[-1, 1]}\left(t\right)\\
        & B^2\left(t\right) =  \left(\frac{3}{4} - t^2\right)\mathbf{1}_{[-\frac{1}{2}, \frac{1}{2}]}\left(t\right) + \frac{1}{2}\left(\frac{3}{2} - |t|\right)^2\mathbf{1}_{[\frac{1}{2}, \frac{3}{2}]}\left(|t|\right)\\
        & B^3\left(t\right) =\frac{1}{6}\left(4 - 3t^2\left(2 - |t|\right)\right)\mathbf{1}_{[0, 1]}\left(|t|\right) + \frac{1}{6}\left(2 - |t|\right)^3\mathbf{1}_{[1, 2]}\left(|t|\right)\\
    \end{split}
\end{equation*}
Using those notations, in equation (7) we define a 3D divergence-conforming SVF $v$ of order $k\geq 2$ and parameters $\left\{(\phi^X_{i}, \phi^Y_{i}, \phi^Z_{i})\right\}_{i \in \{-k, \ldots, n-1\}^3} \in \mathds{R}^{3(n+k)^3}$ as, for all $(x,y,z) \in \Omega$,
\begin{equation*}
    \begin{split}
    v(x,y,z) &= 
    \left(
    \begin{array}{c}
         v^X(x,y,z)  \\
         v^Y(x,y,z) \\
         v^Z(x,y,z)
    \end{array}
    \right)\\
    &= 
    \left(
    \begin{array}{c}
         \sum_{i_X=-k}^{n-1} \sum_{i_Y,i_Z=-(k-1)}^{n-1} B^{k}_{i_X,X}(x) B^{k-1}_{i_Y,Y}(y) B^{k-1}_{i_Z,Z}(z)\,\phi^X_{i_X,i_Y,i_Z}  \\
         \sum_{i_Y=-k}^{n-1} \sum_{i_X,i_Z=-(k-1)}^{n-1} B^{k-1}_{i_X,X}(x) B^{k}_{i_Y,Y}(y) B^{k-1}_{i_Z,Z}(z)\,\phi^Y_{i_X,i_Y,i_Z} \\
         \sum_{i_Z=-k}^{n-1} \sum_{i_X,i_Y=-(k-1)}^{n-1} B^{k-1}_{i_X,X}(x) B^{k-1}_{i_Y,Y}(y) B^{k}_{i_Z,Z}(z)\,\phi^Z_{i_X,i_Y,i_Z}
    \end{array}
    \right)
    \end{split}
\end{equation*}

\subsection*{Remark: Support of B-splines basis functions}
The support of a function $f:\mathcal{X} \mapsto \mathds{R}$ is the set of points where $f$ is non-zero, i.e. $\supp\left(f\right) = \left\{x \in \mathcal{X} \,|\, f(x) \neq 0\right\}$.
For any direction $U\in \{X,Y,Z\}$, for $i\in \{-k,\ldots,n-1\}^3$, the support of the function $B^k_{i,U}$ is $]u_i,\, u_i + (k+1)\delta u[$.

\section{Proof of Lemma 1}
In this subsection, we give more details on the proof of Lemma 1. 
We start by showing equation (5).
The functions $B^k$ have the property:
\begin{equation*}
   \forall k \geq 2,\,\, \forall t \in \mathds{R},\quad \dv{B^k}{t} (t) = B^{k-1}\left(t + \frac{1}{2}\right) - B^{k-1}\left(t - \frac{1}{2}\right)
\end{equation*}

\noindent So using the chain rule, we obtain for all $k\geq 2$ and $i = (i_X, i_Y, i_Z) \in \{1, \ldots, N\}^3$,
\begin{equation*}
    \begin{split}
        \dv{B^k_{i,X}}{x} &=\frac{B^{k-1}\left(\frac{x - x_{i_X}}{\delta x} - \frac{k+1}{2} + \frac{1}{2}\right) - B^{k-1}\left(\frac{x - x_{i_X}}{\delta x} - \frac{k+1}{2} - \frac{1}{2}\right)}{\delta x}\\
                    &=\frac{B^{k-1}\left(\frac{x - x_{i_X}}{\delta x} - \frac{k-1+1}{2}\right) - B^{k-1}\left(\frac{x - (x_{i_X}+\delta x)}{\delta x} - \frac{k-1+1}{2}\right)}{\delta x}\\
                    &= \frac{B^{k-1}_{i,X} - B^{k-1}_{(i_X+1,i_Y,i_Z),X}}{\delta x}\\
    \end{split}
\end{equation*}
\noindent Similarly, for the directions $X,Y$ we obtain:
\begin{equation*}
    \begin{split}
        \dv{B^k_{i,Y}}{y} =& \frac{B^{k-1}_{i,Y} - B^{k-1}_{(i_X,i_Y+1,i_Z),Y}}{\delta y}\\
        \dv{B^k_{i,Z}}{z} =& \frac{B^{k-1}_{i,Z} - B^{k-1}_{(i_X, i_Y,i_Z+1),Z}}{\delta z}\\
    \end{split}
\end{equation*}

\noindent Thus, if $k \geq 2$, for all $(x,y,z) \in \Omega=[0,1]^3$:
\begin{equation*}
    \begin{split}
        \pdv{v^X}{x} = \sum_{i_X=-k}^{n-1} \sum_{i_Y,i_Z=-(k-1)}^{n-1} \frac{B^{k-1}_{i_X,X}(x) - B^{k-1}_{i_X+1,X}(x)}{\delta x} B^{k-1}_{i_Y,Y}(y) B^{k-1}_{i_Z,Z}(z)\,\phi^X_{i_X,i_Y,i_Z}\\
    \end{split}
\end{equation*}
Separating the sum into two sums and using the change of index $i_X := i_X + 1$ in the second sum, we obtain:
\begin{equation*}
    \begin{split}
        \pdv{v^X}{x} =& \sum_{i_X=-k}^{n-1} \sum_{i_Y,i_Z=-(k-1)}^{n-1} B^{k-1}_{i_X,X}(x) B^{k-1}_{i_Y,Y}(y) B^{k-1}_{i_Z,Z}(z)\,\frac{\phi^X_{i_X,i_Y,i_Z}}{\delta x}\\
                &- \sum_{i_X=-(k-1)}^{n} \sum_{i_Y,i_Z=-(k-1)}^{n-1} B^{k-1}_{i_X,X}(x) B^{k-1}_{i_Y,Y}(y) B^{k-1}_{i_Z,Z}(z)\,\frac{\phi^X_{i_X-1,i_Y,i_Z}}{\delta x}\\
    \end{split}
\end{equation*}
We have $B^{k-1}_{-k}(x) = B^{k-1}_n(x) = 0$ for all $x \in [0, 1]$ (see remark in \textbf{1.1}), so we can group the two sums and we finally obtain:
\begin{equation*}
    \begin{split}
        \pdv{v^X}{x} =& \sum_{i_X,i_Y,i_Z=-(k-1)}^{n-1} B^{k-1}_{i_X,X}(x) B^{k-1}_{i_Y,Y}(y) B^{k-1}_{i_Z,Z}(z)\,\frac{\phi^X_{i_X,i_Y,i_Z}-\phi^X_{i_X-1,i_Y,i_Z}}{\delta x}\\
    \end{split}
\end{equation*}

\noindent Similarly we obtain, for all $(x,y,z) \in \Omega=[0,1]^3$:
\begin{equation*}
    \begin{split}
        \pdv{v^Y}{y} =& \sum_{i_X,i_Y,i_Z=-(k-1)}^{n-1} B^{k-1}_{i_X,X}(x) B^{k-1}_{i_Y,Y}(y) B^{k-1}_{i_Z,Z}(z)\,\frac{\phi^Y_{i_X,i_Y,i_Z}-\phi^Y_{i_X,i_Y-1,i_Z}}{\delta y}\\
        \pdv{v^Z}{z} =& \sum_{i_X,i_Y,i_Z=-(k-1)}^{n-1} B^{k-1}_{i_X,X}(x) B^{k-1}_{i_Y,Y}(y) B^{k-1}_{i_Z,Z}(z)\,\frac{\phi^Z_{i_X,i_Y,i_Z}-\phi^Z_{i_X,i_Y,i_Z-1}}{\delta z}\\
    \end{split}
\end{equation*}

\noindent As a result, for all $(x,y,z) \in \Omega=[0,1]^3$, the divergence of $v$ at $(x,y,z)$ is given by:
\begin{equation*}
    \begin{split}
        \div{v}(x,y,z) =& \sum_{i_X,i_Y,i_Z=-(k-1)}^{n-1} B^{k-1}_{i_X,X}(x) B^{k-1}_{i_Y,Y}(y) B^{k-1}_{i_Z,Z}(z)\,\psi_{i_X,i_Y,i_Z}\\
        \text{s.t.}\quad \forall i_X,i_Y,i_Z,\quad \psi_{i_X,i_Y,i_Z} =& \quad
    \frac{\phi^X_{i_X,i_Y,i_Z} - \phi^X_{i_X-1,i_Y,i_Z}}{\delta x} \\
    &+ \frac{\phi^Y_{i_X,i_Y,i_Z} - \phi^Y_{i_X,i_Y-1,i_Z}}{\delta y}\\
    &+ \frac{\phi^Z_{i_X,i_Y,i_Z} - \phi^Z_{i_X,i_Y,i_Z-1}}{\delta z}
    \end{split}
\end{equation*}

\noindent We are now ready to prove Lemma 1.
Let us assume that $k\geq 2$, and $\mathcal{M}$ is a non-empty subset of $\Omega$. Let $\epsilon \geq 0$
and 
\begin{equation*}
\begin{split}
    \mathbf{J}_{\mathcal{M}}=\{ & (i_X,i_Y,i_Z) \in\{-(k-1),\ldots,n-1\}^3
    \,\,|\,\, \\
    & \left(
    \supp B^{k-1}_{i_X,X} \times \supp B^{k-1}_{i_Y,Y} \times \supp B^{k-1}_{i_Z,Z}
    \right)
    \cap \mathcal{M} \neq \emptyset
    \}
\end{split}
\end{equation*}
%
%
\noindent $\mathbf{J}_{\mathcal{M}}$ contains all the indices $(i_X, i_Y, i_Z)$ so that the support of the function $(x,y,z) \mapsto B^{k-1}_{i_X,X}(x) B^{k-1}_{i_Y,Y}(y) B^{k-1}_{i_Z,Z}(z)$ is non-zero for at least one point of $\mathcal{M}$. \medskip

\noindent As a result, if for all $(i_X,i_Y,i_Z) \in \mathbf{J}_{\mathcal{M}},$
\[
    -\epsilon \leq 
    \frac{\phi^X_{i_X,i_Y,i_Z} - \phi^X_{i_X-1,i_Y,i_Z}}{\delta x} +
    \frac{\phi^Y_{i_X,i_Y,i_Z} - \phi^Y_{i_X,i_Y-1,i_Z}}{\delta y} +
    \frac{\phi^Z_{i_X,i_Y,i_Z} - \phi^Z_{i_X,i_Y,i_Z-1}}{\delta z}
    \leq \epsilon,
\]
\noindent Then, for all $(x,y,z) \in \mathcal{M}$,
\begin{equation*}
    \begin{split}
        \abs{\div{v}(x,y,z)} &= \abs{\sum_{i_X,i_Y,i_Z=-(k-1)}^{n-1} B^{k-1}_{i_X,X}(x) B^{k-1}_{i_Y,Y}(y) B^{k-1}_{i_Z,Z}(z)\,\psi_{i_X,i_Y,i_Z}}\\
        &= \abs{\sum_{(i_X,i_Y,i_Z) \in \mathbf{J}_{\mathcal{M}}} B^{k-1}_{i_X,X}(x) B^{k-1}_{i_Y,Y}(y) B^{k-1}_{i_Z,Z}(z)\,\psi_{i_X,i_Y,i_Z}}\\
        &\leq \sum_{(i_X,i_Y,i_Z) \in \mathbf{J}_{\mathcal{M}}} B^{k-1}_{i_X,X}(x) B^{k-1}_{i_Y,Y}(y) B^{k-1}_{i_Z,Z}(z)\,\abs{\psi_{i_X,i_Y,i_Z}}\\
        &\leq \sum_{(i_X,i_Y,i_Z) \in \mathbf{J}_{\mathcal{M}}} B^{k-1}_{i_X,X}(x) B^{k-1}_{i_Y,Y}(y) B^{k-1}_{i_Z,Z}(z)\,\epsilon\\
        &\leq \,\epsilon
    \end{split}
\end{equation*}
Which concludes the proof of Lemma 1.

\section{Proof for Lie exponential approximation and incompressibility}

\noindent In this subsection, we give a detailed proof for the error of incompressibility with an Euler approximation of the Lie exponential (8).
\noindent Let $K \in \mathds{N}$, let $\tau = \frac{1}{2^K}$ be the time step in the Euler integration.
The Euler approximation of the Lie exponential for the time step $\tau$ is given by:
\begin{equation*}
    \begin{split}
        \widetilde{\exp} &= \left(I + \frac{1}{2^K}v\right) \circ \ldots \circ \left(I + \frac{1}{2^K}v\right)\\
            &= \left(I + \frac{1}{2^K}v\right)^{2^K}
    \end{split}
\end{equation*}
where $I$ is the identity mapping.

\noindent Let $\mathcal{M}$ be a non-empty subregion of the spatial domain $\Omega$ and $m=(m_X,m_Y,m_Z) \in \mathcal{M}$.
The Jacobian of the first step of the Euler integration $\left(I + \frac{1}{2^K}v\right)$ at $m$ is given by:
\begin{equation*}
    \begin{split}
        \det\left(J_{I + \frac{1}{2^K} v}(m)\right) &= 
            \left|
            \begin{array}{ccc}
                \frac{1}{2^K}\pdv{v^X}{x}(m) + 1 & \frac{1}{2^K}\pdv{v^X}{y}(m)     & \frac{1}{2^K}\pdv{v^X}{z}(m)\\
                \frac{1}{2^K}\pdv{v^Y}{x}(m)     & \frac{1}{2^K}\pdv{v^Y}{y}(m) + 1 & \frac{1}{2^K}\pdv{v^Y}{z}(m)\\
                \frac{1}{2^K}\pdv{v^Z}{x}(m)     & \frac{1}{2^K}\pdv{v^Z}{y}(m)     & \frac{1}{2^K}\pdv{v^Z}{z}(m) + 1\\
            \end{array}
            \right|\\
            &=1 + \frac{1}{2^K} \div{v}(m) + \mathcal{O}\left(\left(\frac{1}{2^K}\right)^2\right)
    \end{split}
\end{equation*}
Let us note:
\begin{equation*}
\left\{
    \begin{split}
        m_0 &= m\\
    \forall k \in \{1, \ldots, 2^K\}, \quad m_k &= \left(I + \frac{1}{2^K}v\right)(m_{k-1})\\
        &= \left(I + \frac{1}{2^K}v\right)^{k}(m)
    \end{split}
\right.
\end{equation*}
Using the chain rule and the fact that the determinant of the composition of two matrices is equal to the product of their determinant, we obtain:
\begin{equation*}
    \begin{split}
        \det\left(J_{\widetilde{exp}(v)}(m)\right) &= \det\left(J_{\left(I + \frac{1}{2^K} v\right)^{2^K}}(m)\right)\\
        &= \det\left(J_{\left(I + \frac{1}{2^K} v\right) \circ \left(I + \frac{1}{2^K} v\right)^{2^K - 1}}(m)\right)\\
        &= \det\left(J_{\left(I + \frac{1}{2^K} v\right)}(m_{2^K - 1})\right)\det\left(J_{\left(I + \frac{1}{2^K} v\right)^{2^K - 1}}(m)\right)\\
        &= \prod_{k=0}^{2^K-1} \det\left(J_{I + \frac{1}{2^K} v}(m_k)\right)\\
        &= \prod_{k=0}^{2^K-1} \left(1 + \frac{1}{2^K} \div{v}(m_k) + \mathcal{O}\left(\left(\frac{1}{2^K}\right)^2\right)\right)\\
        &= 1 + \frac{1}{2^K}\left(\sum_{k=0}^{2^K-1} \div{v}(m_k)\right) + \mathcal{O}\left(\frac{1}{2^K}\right)
    \end{split}
\end{equation*}

\noindent Let us now assume that $v$ satisfies the condition of Lemma 1 for $\epsilon$ close to $0$, up to machine precision (typically $\epsilon=\epsilon_{mach}=10^{-16}$).
If in addition, $\forall k \in \{0, \ldots, 2^K - 1\}, m_k \in \mathcal{M}$.
Then, $\forall k \in \{0, \ldots, 2^K - 1\}, \abs{\div{v}(m_k)} \leq \epsilon_{mach}$, and
\begin{equation}
    \label{eq:ineq_error}
    \abs{\frac{1}{2^K}\left(\sum_{k=0}^{2^K-1} \div{v}(m_k)\right)} \leq \frac{1}{2^K}\left(\sum_{k=0}^{2^K-1} \abs{\div{v}(m_k)}\right) \leq \epsilon_{mach}
\end{equation}

\noindent As a result, we obtain the approximation given in (8):
\begin{equation*}
    \det\left(J_{\widetilde{exp}(v)}(m)\right) = 1 + \mathcal{O}\left(\tau + \epsilon_{mach}\right)
\end{equation*}

\subsection*{Remark: divergence-free SVF does not guarantee incompressibility of the transformation when $\mathcal{M}$ is local and the deformation is large}
In the case at least one of the $m_k$ is outside of $\mathcal{M}$, inequalities \eqref{eq:ineq_error} do not hold anymore.
Although the divergence of the SVF $v$ is uniformly close to $0$ on $\mathcal{M}$, some point can be deformed in a compressible manner.
This happens when the deformation is large and points that were initially inside $\mathcal{M}$ end up outside of $\mathcal{M}$ during the Euler integration.
In practice, deviations from an incompressible deformation in $\mathcal{M}$ can be observe for accurate divergence-free SVFs in $\mathcal{M}$, as illustrated in Fig. \ref{fig:ssfp_boxplot}.

Fortunately, the definition of $J_{\mathcal{M}}$ implies the presence of margins around the incompressible region $\mathcal{M}$.
As a consequence, inequalities \eqref{eq:ineq_error} are verified in practice for small and moderate deformations in $\mathcal{M}$, as illustrated in Fig. \ref{fig:logJacobian}.
One can verify that thoses margins are linear in the order $k$ of the B-splines basis.

It is worth noting that this limitation is due to the use of SVFs, also used in previous works.
In addition, when $\mathcal{M}$ is equal to the entire image domain $\Omega$, inequalities \eqref{eq:ineq_error} are always satisfied for a divergence-free SVF obtained by our method.
%
%
%
%
\end{document}